\documentclass{article}
\usepackage{graphicx} % Required for inserting images

\usepackage{amsmath,amssymb}
\usepackage{physics}
\usepackage{amsthm}
\usepackage{bm}
\usepackage[shortlabels]{enumitem}
\usepackage[hyphens]{url}

\usepackage{algcompatible}
\usepackage{algorithm}
\usepackage{algpseudocode}
\usepackage[short]{optidef}

\usepackage{mathtools}
\mathtoolsset{showonlyrefs=true}

\newtheorem{theorem}{Theorem}[section]

\newtheorem{lemma}[theorem]{Lemma}
\newtheorem{definition}{Definition}[section]

\newcommand{\R}{\mathbb{R}}

\newcommand{\e}{\varepsilon}
\newcommand{\ri}{\mathrm{ri}}
\newcommand{\T}{\top}
\newcommand{\X}{\mathcal{X}}
\newcommand{\Y}{\mathcal{Y}}
\newcommand{\A}{\mathcal{A}}
\newcommand{\emp}{\varnothing}
\DeclareMathOperator*{\argmax}{arg\,max}
\DeclareMathOperator*{\argmin}{arg\,min}

\title{Convergence analysis and acceleration of the smoothing methods for solving extensive-form games}
\author{Keigo Habara, Ellen Hidemi Fukuda, Nobuo Yamashita}
\date{March 2023}

\begin{document}

\maketitle

\abstract{%
    The extensive-form game has been studied considerably in recent years.
It can represent games with multiple decision points and incomplete information, and hence it is helpful in formulating games with uncertain inputs, such as poker.
We consider an extended-form game with two players and zero-sum, i.e., the sum of their payoffs is always zero.
In such games, the problem of finding the optimal strategy can be formulated as a bilinear saddle-point problem.
This formulation grows huge depending on the size of the game, since it has variables representing the strategies at all decision points for each player.
To solve such large-scale bilinear saddle-point problems, the excessive gap technique (EGT), a smoothing method, has been studied.
This method generates a sequence of approximate solutions whose error is guaranteed to converge at $\order{1/k}$, where $k$ is the number of iterations.
However, it has the disadvantage of having poor theoretical bounds on the error related to the game size.
This makes it inapplicable to large games.

Our goal is to improve the smoothing method for solving extensive-form games so that it can be applied to large-scale games.
To this end, we make two contributions in this work.
First, we slightly modify the strongly convex function used in the smoothing method in order to improve the theoretical bounds related to the game size.
Second, we propose a heuristic called centering trick, which allows the smoothing method to be combined with other methods and consequently accelerates the convergence in practice.
As a result, we combine EGT with CFR+, a state-of-the-art method for extensive-form games, to achieve good performance in games where conventional smoothing methods do not perform well.
The proposed smoothing method is shown to have the potential to solve large games in practice.
}

\section{Introduction}
This study concerns improving the smoothing method for solving extensive-form games with imperfect information.
An \textit{extensive-form game} is a description of a game played by multiple players, where the game state corresponds to a node in a rooted tree.
The state of the game moves down the tree according to the actions of each player, and each player receives a gain when a leaf of the tree is reached.
This game is a good model for games with information gaps between players, such as poker, because it can represent the incomplete information.
We only deal with games in which there are two players, and the sum of their payoff is zero. 
It is known that in such two-person zero-sum games, the problem of finding the optimal strategy can be formulated as a simple minimax problem since the players do not cooperate with each other.

There are two well-known methods for solving extensive-form games: \textit{counterfactual regret minimization} (CFR)~\cite{zinkevich2007regret} and \textit{excessive gap technique} (EGT)~\cite{nesterov2005excessive}.
CFR is an application of regret minimization, a framework used in online learning, to extensive-form games.
Its variant, CFR+~\cite{tammelin2014solving}, was developed to solve a variant of a two-player poker game called Texas Hold'em, which was a challenging problem in artificial intelligence~\cite{bowling2015heads, brown2018superhuman}.
CFR+ is suited for analyzing large games because the solution error is bounded linearly with the game size.
However, it is only guaranteed to converge at $\order{1/\sqrt{k}}$ rate, where $k$ is the number of iterations.
EGT is a smoothing method for the bilinear saddle-point problem, and its application to extensive-form games has been studied~\cite{hoda2010smoothing}.
EGT is theoretically guaranteed to converge with rate $\order{1/k}$, which is faster than CFR+, but it is not suitable for solving large games.
This is because the term $D/\sigma$, which appears in the error bound, depends poorly on the game size.
Here $\sigma$ is the strong convexity parameter of a function $d$ called \textit{prox-function} used in EGT and $D:=\max_{\bm x} d(\bm x)$.
Various prox-functions have been proposed in the literature~\cite{hoda2010smoothing,kroer2020faster,farina2021better}, with $D/\sigma$ being bounded by a cubic order of the game size.

Our goal is to improve the prox-function so that EGT can be applied to large extensive-form games.
To this purpose, we make two contributions.
First, we improve the bound of $D/\sigma$ from $M_Q^3\ln\max_{I\in\mathcal{I}}\abs{\A(I)}$ to $M_Q^2\ln\abs{\Sigma}$ by eliminating the first-order term of the prox-function proposed in~\cite{farina2021better}.
Here $M_Q$ is the maximum value of the $L_1$-norm among the feasible points of the bilinear saddle-point problem, and $\abs{\Sigma}$ is the dimension of the feasible set, both of which are bounded by the game size.
Also, $\max_{I\in\mathcal{I}}\abs{\A(I)}$ is the maximum number of legal actions at each decision point.
Second, we propose a heuristic in EGT that we call the centering trick.
This trick modifies the prox-function in a way that it takes the minimum value in a temporary solution. 
It is expected to improve the accuracy of the smooth approximation of the bilinear saddle-point problem inside the EGT.
This heuristic accelerates convergence in practice and can be combined with other methods by using their solutions.
Numerical experiments show that EGT with the heuristic combined with CFR+ performs best among several methods, including CFR+, for games of a scale where conventional EGT does not perform well.
This suggests that the proposed smoothing method is effective even for large games.
%This suggests that smoothing methods may be helpful even for large games.

The structure of this paper is given as follows.
Section \ref{sec-pre} introduces the basics of convexity analysis.
Section \ref{sec-smoothing} presents the bilinear saddle-point problem, the prox-function, and the smoothing method (EGT).
Section \ref{sec-efg} introduces the extensive-form games and explains how to transform the problem of finding the optimal strategy of the game into a problem covered by EGT.
In Section \ref{sec-prox}, we propose a prox-function, and we will show that it improves theoretical convergence.
Section \ref{sec-experiments} provides the centering heuristic that accelerates the convergence of EGT in practice and confirms its performance through numerical experiments.
Section \ref{sec-conclusions} summarizes our contributions.
\section{Preliminaries}
\label{sec-pre}

This section introduces the basic properties of the convex analysis used in this paper.
Let us denote $\R^n$ as the $n$-dimentional Euclidean space and $\ri S$ as the relative interior of $S\subset\R^n$. 
We will denote $\bm x^\T$ and $\bm A^\T$ as the transpose of the $n$-dimentional vector $\bm x\in\R^n$ and the $n\times m$ matrix $\bm A\in\R^{n\times m}$.
For a convex compact set $S\subset\R^n$, let us define the conjugate function $f^*\colon\R^n\to\R$ of a function $f\colon S\to\R$ by
\begin{align}
    f^*(\bm\xi) := \max_{\bm x\in S}\qty{
    \bm\xi^\T \bm x - f(\bm x)
    }.
\end{align}
If the above maximizer is unique, then by Danskin's theorem, $f^*$ is differentiable at $\bm\xi\in\R^n$, and its derivative is given by
\begin{align}
    \nabla f^*(\bm\xi) = \argmax_{\bm x\in S}\qty{
    \bm\xi^\T \bm x - f(\bm x)
    }.
\end{align}
\section{Smoothing methods}
\label{sec-smoothing}
In this section, the bilinear saddle-point problem and the prox-function are explained, followed by an overview of EGT.

\subsection{Bilinear saddle-point problems}
The bilinear saddle-point problems (BSPPs), which EGT covers, are written in the following form:
\begin{align}
    \min_{\bm x\in\X} \max_{\bm y\in\Y} \bm x^\T \bm A \bm y,
    \label{eqn-bspp}
\end{align}
where $\bm A\in\R^{n\times m}$ is a matrix and $\X \subset \R^n, \Y \subset \R^m$ are convex and compact sets.
The \textit{adjoint form} of the problem is given by
\begin{align}
    \max_{\bm y\in\Y} \min_{\bm x\in\X} \bm x^\T \bm A \bm y.
    \label{eqn-bspp-adjoint}
\end{align}
By Minimax Theorem, these two problems achieve the same optimal value.
In other words, the following equation holds for the optimal solution $\bm x^*\in\X$ and $\bm y^*\in\Y$ of~\eqref{eqn-bspp} and~\eqref{eqn-bspp-adjoint}:
\begin{align}
    \max_{\bm y\in\Y} (\bm x^*)^\T \bm A \bm y
    = 
    \min_{\bm x\in\X} \bm x^\T \bm A \bm y^*.
\end{align}
The error of the pair of the solutions $(\bar{\bm x}, \bar{\bm y})$ can be defined by 
\begin{align}
    \e(\bar{\bm x}, \bar{\bm y}) 
    := \max_{\bm y\in\Y} \bar{\bm x}^\T \bm A \bm y 
    - \min_{\bm x\in\X} \bm x^\T \bm A \bar{\bm y}
    \ge 0,
    \label{eqn-error}
\end{align}
and $\e(\bar{\bm x},\bar{\bm y}) = 0$ if and only if $\bar{\bm x}$ and $\bar{\bm y}$ are the optimal solutions of problems~\eqref{eqn-bspp} and~\eqref{eqn-bspp-adjoint}.

\subsection{Prox-functions}
BSPPs~\eqref{eqn-bspp}, the problems covered by EGT, can be regarded as non-smooth minimization problems because $f(\bm x):= \max_{\bm y\in\Y} \bm x^\T \bm A \bm y$ is generally non-smooth.
One way to $f$ is to consider a \textit{prox-function} which we define below.
\begin{definition}
\label{def-prox}
A function $d\colon S\to\R$ is called a prox-function on a convex compact set $S\subset\R^n$ when satisfying the following conditions:
\begin{itemize}
    \item $d$ is twice differentiable in $\ri S$;
    \item $d$ is $\sigma$-strongly convex in $\ri S$ with respect to some norm $\norm{\cdot}_S$ defined on $\R^n$, i.e.,
    \begin{align}
        \bm\xi^\T \nabla^2 d(\bm x) \bm\xi
        \ge \sigma\norm{\bm\xi}_S^2 \quad \forall \bm x\in\ri S,
        \forall \bm\xi\in\R^n;
    \end{align}
    \item $\min_{\bm x\in S} d(\bm x) = 0$.
\end{itemize}
\end{definition}
For some prox-function $d_2\colon\Y\to\R$ and parameter $\mu_2 > 0$, define the \textit{smoothing} of $f$ as
\begin{align}
    f_{\mu_2}(\bm x) 
    &:= \max_{\bm y\in\Y} \qty{
    \bm x^\T \bm A \bm y
    - \mu_2 d_2(\bm y)
    } 
    \label{eqn-f-mu}
    \\
    &= \mu_2 d_2^*\qty(
    \bm A^\T \bm x / \mu_2
    ).
\end{align}
Since $d_2$ is strongly convex, the maximizer of~\eqref{eqn-f-mu} is unique, then $d_2^*$ is differentiable, which in turn shows that $f_{\mu_2}$ is also differentiable.
Let $D_2 := \max_{\bm y\in\Y} d_2(\bm y)$.
By the definition of $f_{\mu_2}$, we have
\begin{align}
    f_{\mu_2}(\bm x) \le f(\bm x) 
    \le f_{\mu_2}(\bm x) + \mu_2 D_2
    \quad \forall \bm x \in \X,
    \label{eqn-f-f-mu}
\end{align}
then $f_{\mu_2}$ is a good smooth approximation of $f$ for small $\mu_2>0$.

\subsection{Excessive gap technique}
This section provides an overview of EGT.
Let $\phi(\bm y) := \min_{\bm x\in\X} \bm x^\T\bm A\bm y$. 
For some prox-function $d_1\colon\X\to\R$ and $\mu_1 > 0$, define the smooth approximation of $\phi$ similarly to~\eqref{eqn-f-mu}:
\begin{align}
    \phi_{\mu_1}(\bm y) 
    &:= \min_{\bm x\in\X}\qty{
    \bm x^\T\bm A\bm y + \mu_1 d_1(\bm x)
    },
\end{align}
and we also have 
\begin{align}
    \phi_{\mu_1}(\bm y) - \mu_1 D_1 \le \phi(\bm y) \quad \forall \bm y\in\Y,
    \label{eqn-phi-phi-mu}
\end{align}
where $D_1 := \max_{\bm x\in\X} d_1(\bm x)$.
Here we consider the following condition, called \textit{excessive gap condition}:
\begin{align}
    f_{\mu_2}(\bm x) \le \phi_{\mu_1}(\bm y).
    \label{eqn-egc}
\end{align}
If $(\bm x,\bm y)\in\X\times\Y$ and $\mu_1,\mu_2 > 0$ satisfy the excessive gap condition, the error of $(\bm x, \bm y)$, defined in~\eqref{eqn-error}, is bounded by
\begin{align}
    \e\qty(\bm x, \bm y)
    &= f(\bm x) - \phi(\bm y) \\
    &\le f_{\mu_2}(\bm x) - \phi_{\mu_1}(\bm y) + \mu_1D_1 + \mu_2D_2 \\
    &\le \mu_1 D_1 + \mu_2 D_2,
\end{align}
where the first inequality follows from~\eqref{eqn-f-f-mu} and~\eqref{eqn-phi-phi-mu}, and the second inequality is given by the excessive gap condition~\eqref{eqn-egc}.
The central concept of EGT is to find $(\bm x, \bm y)$ such that the excessive gap condition is satisfied for small $\mu_1, \mu_2 > 0$ in order to reduce $\e(\bm x, \bm y)$. To achieve this, the EGT algorithm consists of two parts: \textit{initialization} and \textit{shrinking}.
%The central concept of EGT is to obtain $(\bm x, \bm y)$ such that the excessive gap condition is satisfied with small $\mu_1,\mu_2>0$ to make $\e(\bm x, \bm y)$ small. For this purpose, the algorithm of EGT consists of two parts: \textit{initialization} and \textit{shrinking}.

\begin{itemize}
    \item The \textit{initialization} part: Algorithm~\ref{alg-init} requires $\mu_1,\mu_2>0$ satisfying $\mu_1\mu_2 \ge \norm{\bm A}^2/(\sigma_1\sigma_2)$, where 
    \begin{align}
        \norm{\bm A} := 
        \max_{\norm{\bm x}_\X=1}
        \max_{\norm{\bm y}_\Y=1}
        \bm x^\T \bm A \bm y
    \end{align}
    and $d_1, d_2$ must be $\sigma_1, \sigma_2$-strongly convex with respect to some norm $\norm{\cdot}_\X, \norm{\cdot}_\Y$, respectively.
    Then it generates $\qty(\bm x^0,\bm y^0)$ satisfying the excessive gap condition with its input $\mu_1,\mu_2$.

    \item The \textit{shrinking} part: Algorithm~\ref{alg-shrink} requires $\qty(\bm x^k,\bm y^k)$ and $\mu_1,\mu_2>0$ satisfying the excessive gap condition, and the step size $\tau\in(0,1)$ with $\tau^2/(1-\tau) \le \sigma_1\sigma_2\mu_1\mu_2/\norm{\bm A}^2$.
    Then it generates $\qty(\bm x^{k+1},\bm y^{k+1})$ satisfying the excessive gap condition with $(1-\tau)\mu_1,\mu_2$, i.e. we can shrink $\mu_1$ by a factor of $1-\tau$.
    The shrinking algorithm of the $\mu_2$ version can be obtained similarly.
\end{itemize}

By using Algorithms~\ref{alg-init} and~\ref{alg-shrink} with the parameters and choices between $\mu_1$ and $\mu_2$ to shrink presented in~\cite{nesterov2005excessive}, the solution $\qty(\bm x^k,\bm y^k)$ generated by EGT guarantees the following convergence result~\cite[Theorem 6.3]{nesterov2005excessive}:
\begin{align}
    \e(\bm x^k,\bm y^k) \le 
    \frac{4\norm{\bm A}}{k+1}\sqrt{\frac{D_1D_2}{\sigma_1\sigma_2}}.
    \label{eqn-conv}
\end{align}
Therefore, a small $\norm{\bm A}$ and large values for $\sigma_1/D_1, \sigma_2/D_2$ which we call \textit{substantial strongly convexity} of $d_1, d_2$, are required for better convergence, with some norm $\norm{\cdot}_\X, \norm{\cdot}_\Y$.
In this paper, we only consider the $L_1$-norm, which is desirable because $\norm{\bm A} = \max_{i,j} \abs{A_{ij}}$ does not depend on the dimensions of $\X$ and $\Y$.
In practice, we also require that $\nabla d_1^*$ and $\nabla d_2^*$ are easily computable.

\begin{algorithm}[H]
\caption{Initialization}
\label{alg-init}
\begin{algorithmic}[1]
\Require $\mu_1, \mu_2 > 0$ satisfying $\mu_1\mu_2 \ge \norm{\bm A}^2/(\sigma_1\sigma_2)$
\Ensure $\bm x^0, \bm y^0$ satisfying the excessive gap condition~\eqref{eqn-egc} with $\mu_1, \mu_2$
\State $\bar{\bm x} \gets \nabla d_1^*(\bm 0)$
\State $\bm y^0 \gets \nabla d_2^*(\bm A^\T\bar{\bm x} / \mu_2)$
\State $\bm x^0 \gets \nabla d_1^*(\nabla d_1(\bar{\bm x}) - \bm A \bm y^0/\mu_1)$
\end{algorithmic}
\end{algorithm}

\begin{algorithm}[H]
\caption{Shrinking ($\mu_1$ version)}
\label{alg-shrink}
\begin{algorithmic}[1]
\Require
\Statex $\bm x^k, \bm y^k, \mu_1, \mu_2$ satisfying the excessive gap condition~\eqref{eqn-egc}
\Statex $\tau \in (0, 1)$ satisfying $\tau^2/(1-\tau) \le \sigma_1\sigma_2\mu_1\mu_2/\norm{\bm A}^2$
\Ensure
\Statex $\bm x^{k+1}, \bm y^{k+1}$ satisfying the excessive gap condition with $(1-\tau)\mu_1, \mu_2$
\State $\bar{\bm x} \gets \nabla d_1^*(-\bm A \bm y^k / \mu_1)$
\State $\bar{\bm y} \gets \nabla d_2^*(\bm A^\T ((1-\tau) \bm x^k + \tau \bar{\bm x}) / \mu_2)$
\State $\bm y^{k+1} \gets (1-\tau) \bm y^k + \tau \bar{\bm y}$
\State $\bm x^{k+1} \gets (1-\tau) \bm x^k + \tau \nabla d_1^*(\nabla d_1(\bar{\bm x}) - \frac{\tau}{(1-\tau)\mu_1}\bm A\bar{\bm y})$
\end{algorithmic}
\end{algorithm}
\section{Extensive-form games}
\label{sec-efg}

An extensive-form game is a game played by $N$ players, which can be represented with a rooted tree.
A state of the game corresponds to a node of the tree.
Moving down the tree can be done either with the player's actions or stochastic events (like dealing cards, etc.) and hereafter, such stochastic events will be considered as \textit{chance player}'s actions.
The game ends when a leaf of the tree is reached, and player $i=1,\dots, N$ receives a gain $u_i(z)$ corresponding to the terminal node $z$. 
This paper only deals with two-person zero-sum games. 
That is $N=2$ and $u:=-u_1=u_2$, where $u$ is a loss for player 1 and a gain for player 2.

Let $H_i$ denote the set of nodes where player $i=1,2$ acts.
Partitions of $H_i$ describe the imperfect information of the game.
Such a partition $\mathcal{I}_i$ is called an \textit{information partition}, and each element $I\in\mathcal{I}_i$ (i.e. $I\subset H_i$) is called an \textit{information set}. 
Player $i=1,2$ cannot distinguish between nodes belonging to the same information set. 
The information partition must satisfy the following natural constraint: all nodes belonging to the same information set must have equal sets of legal actions.
Also, in this paper, we only consider games satisfying \textit{perfect recall}, i.e., the information partition is consistent with the assumption that each player can remember their own past actions.
See Figures~\ref{fig-kuhn} and~\ref{fig-leduc} in Appendix~\ref{sec-app-games} for an extensive-form game representation of Kuhn poker and Leduc Hold'em, simplified versions of Texas Hold'em.

Assume that each player $i=1,2$ can choose their actions probabilistically at each information set. 
Let $\pi_i(z)$ be the contribution of player $i=1,2,c$ to the probability of reaching the terminal node $z$ from the root, where $c$ means the chance player.
Let $Z$ be the set of terminal nodes, then the expected value of $u$ is given by
\begin{align}
    \sum_{z\in Z} \pi_1(z)\pi_2(z)\pi_c(z)u(z).
\end{align}
Each player $i=1,2$ aims to make this expectation smaller and larger by controlling $\pi_1$ and $\pi_2$, respectively.
Note that $\pi_c$ is constant.

Now consider the feasible region of $\pi_i$.
Let
\begin{align}
    \Sigma_i := \qty{\emp} \cup \qty{
    (I, a) \mid I\in\mathcal{I}_i, a\in\A(I)
    },
\end{align}
where $\A(I)\ne\emptyset$ is the set of legal actions at information set $I$.
Furthermore, by the assumption of \textit{perfect recall}, we can define the \textit{parent function} $p_i\colon\mathcal{I}_i\to\Sigma_i$ such that $p_i(I) = (I^\prime, a)$ if and only if $I^\prime\in\mathcal{I}_i$ is the last information set visited before $I\in\mathcal{I}_i$ and the action $a\in\A(I^\prime)$ is chosen there, and $p_i(I) = \emp$ if and only if there is no player $i$'s information set that comes before $I\in\mathcal{I}_i$.
We see that the function $p_i$ forms a tree.
That is, for a set of vertices $\Sigma_i\cup\mathcal{I}_i$, the graph with edges from $p_i(I)$ to $I$ for all $I\in\mathcal{I}_i$ and from $I$ to $(I, a)$ for all $I\in\mathcal{I}_i$ and $a\in\A(I)$ is a tree rooted at $\emp$.
See Figures~\ref{fig-kuhnq1},~\ref{fig-kuhnq2}, and~\ref{fig-leducq1} for this \textit{information trees} in each game.

Then let us consider the following convex compact set:
\begin{align}
    Q_i := \qty{
    \bm x\in\R_{\ge 0}^\abs{\Sigma_i}
    \,\middle\vert\,
    x_{\emp} = 1, 
    \ x_{p_i(I)} = \sum_{a\in\A(I)} x_{I,a}
    \ \forall I\in\mathcal{I}_i
    },
    \label{eqn-strategy}
\end{align}
which we call the \textit{strategy set} for player $i$.
Consider mapping $\bm x\in Q_i$ to a probabilistic strategy that chooses action $a\in\A(I)$ at information set $I\in\mathcal{I}_i$ with the following probability:
\begin{align}
\begin{cases}
    x_{I,a} / x_{p_i(I)}, & \text{if } x_{p_i(I)} \ne 0, \\
    1/\abs{\A(I)}, & \text{otherwise},
\end{cases}
\end{align}
which is a sufficient representation of the player $i$'s probabilistic strategy.
Furthermore, the contribution $\pi_i(z)$ is given by $x_{p_i(z)}$
where $p_i(z)$ is a natural extension of the parent function.
That is, $p_i(z) = (I, a)$ if and only if $I\in\mathcal{I}_i$ is the last information set visited before $z\in Z$ and the action $a\in\A(I)$ is chosen there, and $p_i(z) = \emp$ if and only if there is no player $i$'s information set before $z\in Z$.

Therefore, extensive-form games can be written as the following BSPP:
\begin{align}
    \min_{\bm x\in Q_1}\max_{\bm y\in Q_2}
    \sum_{z\in Z} x_{p_1(z)} y_{p_2(z)} \pi_c(z)u(z).
\end{align}
As can be seen from this formulation, the matrix $\bm A$ in BSPP~\eqref{eqn-bspp}, which represents an extensive-form game, has only at most $\abs{\Sigma}$ non-zero elements. In other words, $\bm A$ is sparse in most cases.
\section{Prox-function over the strategy set}
\label{sec-prox}
For solving extensive-form games with EGT, we need a prox-function defined on the convex compact set $Q_i$, defined by~\eqref{eqn-strategy}, the strategy set of each player $i=1,2$.
To guarantee good convergence, it is necessary to define a prox-function $d_i\colon Q_i\to\R$ with large \textit{substantial strong convexity} $\sigma_i/D_i$.

For simplicity, we omit $i$ denoting the player in this section.
Let $M_{Q} := \max_{\bm x\in Q} \norm{\bm x}_1$, which represents the scale of the game.
The prox-function proposed in a previous study~\cite{farina2021better} is shown to be $1/(M_{Q}^3\ln\max_{I\in\mathcal{I}}\abs{\A(I)})$-strongly convex substantially with respect to the $L_1$-norm.
We propose a slightly modified version of this prox-function and show that it is $1/(M_{Q}^2\ln\abs{\Sigma})$-strongly convex substantially with respect to the $L_1$-norm, which is a better guarantee for most games.

Now we propose the prox-function $d\colon Q\to\R$ defined by
\begin{align}
    d(\bm x) := x_\emp\ln x_\emp 
    + \sum_{I\in\mathcal{I}}\sum_{a\in\A(I)} 
    \qty(w_I - \sum_{p(I^\prime)=(I,a)} w_{I^\prime})
    \ x_{I,a}\ln x_{I,a},
    \label{eqn-prox}
\end{align}
where $w_I\in\R$ is defined recursively:
\begin{align}
    w_I := 1 + \max_{a\in\A(I)} 
    \sum_{p(I^\prime)=(I,a)} w_{I^\prime}
    \quad \forall I\in\mathcal{I},
    \label{eqn-w}
\end{align}
and its base case is given by $I\in\mathcal{I}$ with $\qty{I^\prime\in\mathcal{I}\mid p(I^\prime) = (I, a)} = \emptyset$ for all $a\in\A(I)$.
We will denote $\ln x$ as the natural logarithm of $x\ge0$ and assume $0\ln0=0$.
Note that \eqref{eqn-prox} does not satisfy $\min_{\bm x\in Q} d(\bm x)=0$, the third condition for prox-function (see Definition~\ref{def-prox}), so $d-\min_{\bm x\in Q}d(\bm x)$ must be used instead.
For simplicity, however, we will treat $d$ as a prox-function in the following.
This is because the additional constant term is not essential; it does not affect the strong convexity or $\nabla d^*$ and only shifts $d^*$ by a constant.

Note that the prox-function proposed in~\cite{farina2021better} is given by
\begin{align}
    d(\bm x)
    + \sum_{I\in\mathcal{I}} w_I x_{p(I)} \ln\abs{\A(I)},
    \label{eqn-farina}
\end{align}
which is shown to have a minimum value zero.
We have eliminated the first-order term in~\eqref{eqn-farina} by neglecting the adjustment of the minimum to zero. 
As a result, we succeeded in giving a better theoretical guarantee of substantial strong convexity.

\begin{theorem}
\label{thm-sigma}
The prox-function~\eqref{eqn-prox} is $1/M_Q$-strongly convex with respect to the $L_1$-norm.
\end{theorem}
\begin{proof}
This proof is the same as the proof of~\cite[Theorem 5]{farina2021better}.
For $\bm\xi\in\R^\abs{\Sigma}$ and $\bm x\in\ri Q$, we have
\begin{align}
    \bm\xi^\T\nabla^2 d(\bm x)\bm\xi
    &= \frac{\qty(\xi_\emp)^2}{x_\emp} +
    \sum_{I\in\mathcal{I}}\sum_{a\in\A(I)}\qty(
    w_I - \sum_{p(I^\prime)=(I,a)} w_{I^\prime}
    ) \frac{\qty(\xi_{I,a})^2}{x_{I,a}} \\
    &\ge \frac{\qty(\xi_\emp)^2}{x_\emp} +
    \sum_{I\in\mathcal{I}}\sum_{a\in\A(I)} \frac{\qty(\xi_{I,a})^2}{x_{I,a}} \\
    &= \sum_{j\in\Sigma} \frac{\qty(\xi_j)^2}{x_j} \\
    &\ge \frac{
    \qty(\sum_{j\in\Sigma}\abs{\xi_j})^2
    }{\sum_{j\in\Sigma} x_j} \\
    &= \frac{\norm{\bm\xi}_1^2}{\norm{\bm x}_1} 
    \ge \frac{\norm{\bm\xi}_1^2}{M_Q},
\end{align}
where the first equality follows from~\eqref{eqn-prox}, the second inequality comes from~\eqref{eqn-w}, and the fourth inequality is true from the Cauchy-Schwarz inequality:
\begin{align}
    \sum_{j\in\Sigma} \qty(\sqrt{x_j})^2
    \cdot
    \sum_{j\in\Sigma} \qty(\frac{\abs{\xi_j}}{\sqrt{x_j}})^2
    \ge
    \qty(\sum_{j\in\Sigma} \sqrt{x_j}\cdot\frac{\abs{\xi_j}}{\sqrt{x_j}})^2.
\end{align}
\end{proof}

To consider the properties of the conjugate function of the prox-function $d$, we present the following corollary.
\begin{lemma}
\label{lem-prox-another}
The prox-function~\eqref{eqn-prox} satisfies the following equation for $\bm x \in \ri Q$:
\begin{align}
    d(\bm x) 
    &= \sum_{I\in\mathcal{I}}
    w_I x_{p(I)}
    \sum_{a\in\A(I)} 
    \frac{x_{I,a}}{x_{p(I)}}
    \ln\frac{x_{I,a}}{x_{p(I)}}.
    \label{eqn-prox-another}
\end{align}
\end{lemma}
\begin{proof}
First, note that $x_j \ne 0$ for $j\in\Sigma$ for $\bm x\in\ri Q$.
Then for $\bm x\in\ri Q$, we have
%For $\bm x\in Q$, we have $x_\emp\ln x_\emp = 1\ln 1=0$ and $x_{p(I)}=\sum_{a\in\A(I)} x_{I,a}$.
%Note also that $x_j \ne 0$ for $j\in\Sigma$ for $\bm x\in\ri Q$, and we have
\begin{align}
d(\bm x)
&= \sum_{I\in\mathcal{I}}\sum_{a\in\A(I)} 
   \qty(w_I - \sum_{p(I^\prime)=(I,a)} w_{I^\prime})
   \ x_{I,a}\ln x_{I,a} \\
&= 
\qty(
\sum_{I\in\mathcal{I}}\sum_{a\in\A(I)} w_I x_{I,a}\ln x_{I,a}
) - 
\sum_{p(I^\prime)\ne\emp} w_{I^\prime} x_{p(I^\prime)}\ln x_{p(I^\prime)} 
\\
&= 
\sum_{I\in\mathcal{I}}\sum_{a\in\A(I)} w_I x_{I,a}\ln x_{I,a}
- 
\sum_{I\in\mathcal{I}} w_{I} x_{p(I)}\ln x_{p(I)} 
\\
&= \sum_{I\in\mathcal{I}} w_I \qty{
\qty( \sum_{a\in\A(I)} x_{I,a}\ln x_{I,a} )
 - x_{p(I)}\ln x_{p(I)}
} \\
&= \sum_{I\in\mathcal{I}} w_I \qty{
\sum_{a\in\A(I)} x_{I,a}\ln x_{I,a}
- \sum_{a\in\A(I)} x_{I,a}\ln x_{p(I)}
} \\
&= \sum_{I\in\mathcal{I}} w_I \sum_{a\in\A(I)}
x_{I,a}\ln\frac{x_{I,a}}{x_{p(I)}} \\
&= \sum_{I\in\mathcal{I}}
w_I x_{p(I)}
\sum_{a\in\A(I)} 
\frac{x_{I,a}}{x_{p(I)}}
\ln\frac{x_{I,a}}{x_{p(I)}},
\end{align}
where the first and third equalities follow from $x_\emp\ln x_\emp = 1\ln 1 = 0$ for $\bm x\in Q$, and the fifth equality comes from $x_{p(I)}=\sum_{a\in\A(I)} x_{I,a}$ for $\bm x\in Q$.
\end{proof}
From Lemma~\ref{lem-prox-another} and the fact that~\eqref{eqn-prox} is continuous in $Q$, we have
\begin{align}
    d^*(\bm\xi) 
    &= \max_{\bm x\in Q} \qty{
    \bm\xi^\T \bm x - d(\bm x)
    } \\
    &= \sup_{\bm x\in \ri Q} \qty{
    \bm\xi^\T \bm x - 
    \sum_{I\in\mathcal{I}}
    w_I x_{p(I)}
    \sum_{a\in\A(I)} 
    \frac{x_{I,a}}{x_{p(I)}}
    \ln\frac{x_{I,a}}{x_{p(I)}} 
    }.
    \label{eqn-sup-conj}
\end{align}
%Now, choose any $I\in\mathcal{I}$ such that there is no $I^\prime\in\mathcal{I}$ with $p(I^\prime)=(I,a)$ for all $a\in\A(I)$, then the terms on $\qty(x_{I,a})_{a\in\A(I)}$ of the supreme \eqref{eqn-sup-conj} are given by
Now, choose any $I\in\mathcal{I}$ satisfying $\qty{I^\prime\in\mathcal{I}\mid p(I^\prime)=(I,a)}=\emptyset$ for all $a\in\A(I)$, then the terms on $\qty(x_{I,a})_{a\in\A(I)}$ of the supreme~\eqref{eqn-sup-conj} are given by
\begin{align}
    x_{p(I)} 
    \sup_{\bm z\in\ri\Delta_{\abs{\A(I)}}} \qty{
    \sum_{a\in\A(I)} \xi_{I,a} z_a
    - w_I \sum_{a\in\A(I)} z_a\ln z_a
    },
\end{align}
where $z_a := x_{I,a}/x_{p(I)}$ and $\Delta_n$ is the $n$-dimentional simplex.
This maximization subproblem can be solved analytically (see Appendix~\ref{sec-app-subproblem}).
That is, the maximizer is given by
\begin{align}
z_a^* 
:=
\frac{\exp(\xi_{I,a}/w_I)}{\sum_{a^\prime\in\A(I)}\exp(\xi_{I,a^\prime}/w_I)},
\end{align}
which achieves the following supreme:
\begin{align}
    \mathrm{opt}_I
    &:= \sum_{a\in\A(I)}\xi_{I,a} z_a^* - w_I\sum_{a\in\A(I)} z_a^*\ln z_a^* \\
    &= w_I\ln\qty{
    \sum_{a\in\A(I)}\exp(\xi_{I,a}/w_I)
    }.
\end{align}
Substituting this result to~\eqref{eqn-sup-conj}, the terms on $\qty(x_{I,a})_{a\in\A(I)}$ disappear and $\mathrm{opt}_I$ is added to $\xi_{p(I)}$, which is the coefficient of $x_{p(I)}$ in~\eqref{eqn-sup-conj}.

By repeating the above operations in bottom-up order,~\eqref{eqn-sup-conj} can be solved, and the total calculation can be performed in $\order{\abs{\Sigma}}$.
We can also obtain $\nabla d^*(\bm\xi) = \argmax_{\bm x\in Q}\qty{\bm\xi^\T\bm x - d(\bm x)}$, however, only the ratio $\bm z$ is obtained in the above operations.
Then, after solving~\eqref{eqn-sup-conj}, we need to calculate $\nabla d^*(\bm\xi)$ by multiplying $\bm z$ in top-down order.
See Algorithm~\ref{alg-conj-grad} for details.

\begin{algorithm}[H]
\caption{Calculating $d^*(\bm\xi)$ and $\nabla d^*(\bm\xi)$}
\label{alg-conj-grad}
\begin{algorithmic}[1]
\Require $\bm\xi\in\R^\abs{\Sigma}$
\Ensure $y = d^*(\bm\xi),\ \bm z = \nabla d^*(\bm\xi)$
\State $\bm z \gets \bm 0 \in \R^\abs{\Sigma}$
\For{$I\in\mathcal{I}$ in bottom-up order}
    \For{$a\in\A(I)$}
        \State $\xi_{I,a} \gets \exp(\xi_{I,a} / w_I)$
    \EndFor
    \State $\xi_{p(I)} \gets \xi_{p(I)} + w_I\ln\sum_{a\in\A(I)} \xi_{I,a}$
    \For{$a\in\A(I)$}
        \State $z_{I,a} \gets \xi_{I,a} / \sum_{a\in\A(I)} \xi_{I,a}$
    \EndFor
\EndFor
\State $y \gets \xi_\emp$
\State $z_\emp \gets 1$
\For{$I\in\mathcal{I}$ in top-down order}
    \For{$a\in\A(I)$}
        \State $z_{I,a} \gets z_{p(I)} z_{I,a}$
    \EndFor
\EndFor
\end{algorithmic}
\end{algorithm}

\begin{theorem}
\label{thm-d}
The prox-function~\eqref{eqn-prox} satisfies
\begin{align}
    \max_{\bm x\in Q} d(\bm x) 
    - \min_{\bm x\in Q} d(\bm x)
    \le M_Q \ln\abs{\Sigma}.
\end{align}
\end{theorem}
\begin{proof}
Since $Q\subset[0,1]^\abs{\Sigma}$ and $x\ln x \le 0$ for $x\in[0,1]$, we have $\max_{\bm x\in Q} d(\bm x) \le 0$.
Then it is sufficient to show $-\min_{\bm x\in Q} d(\bm x) = d^*(\bm 0) \le M_Q\ln\abs{\Sigma}$.
From the procedure for computing $d^*(\bm\xi)$ presented above, we see that $d^*(\bm 0)$ satisfies the following recursive equation:
\begin{align}
    d^*(\bm 0) 
    &= \sum_{p(I)=\emp} \mathrm{opt}_I \\
    \mathrm{opt}_I
    &= w_I\ln \qty{\sum_{a\in\A(I)} \exp\qty(
    \frac{\sum_{p(I^\prime)=(I,a)} 
    \mathrm{opt}_{I^\prime}}{w_I}
    )
    } \quad \forall I\in\mathcal{I} \label{eqn-opt-i}
\end{align}
Now define $\gamma_I\in\R$ recursively:
\begin{align}
    \gamma_I := \abs{\A(I)} 
    + \sum_{a\in\A(I)}\sum_{p(I^\prime)=(I,a)}
    \gamma_{I^\prime} \quad \forall I\in\mathcal{I},
    \label{eqn-gamma}
\end{align}
then let us show
\begin{align}
    \mathrm{opt}_I \le w_I\ln\gamma_I
    \quad \forall I\in\mathcal{I}
    \label{eqn-opt-ineq}
\end{align}
recursively.
Assume that~\eqref{eqn-opt-ineq} holds for $I^\prime$ which is below $I$ in the sense of the rooted tree, discussed in Section~\ref{sec-efg}.
Then, we have
\begin{align}
    \sum_{p(I^\prime)=(I,a)}
    \mathrm{opt}_{I^\prime}
    &\le \sum_{p(I^\prime)=(I,a)}
    w_{I^\prime}\ln\gamma_{I^\prime} \\
    &\le \qty(
    1 + \sum_{p(I^\prime)=(I,a)} w_{I^\prime}
    )\ln\qty(
    1 + \sum_{p(I^\prime)=(I,a)}\gamma_{I^\prime}
    ) \\
    &\le w_I\ln\qty(
    1 + \sum_{p(I^\prime)=(I,a)}\gamma_{I^\prime}
    ),
\end{align}
where the first inequality follows from the assumption, and the third inequality comes from~\eqref{eqn-w}.
By substituting this to~\eqref{eqn-opt-i}, we have
\begin{align}
    \mathrm{opt}_I
    &\le w_I\ln\qty{
    \sum_{a\in\A(I)} \qty(
    1 + \sum_{p(I^\prime)=(I,a)}
    \gamma_{I^\prime}
    )
    } \\
    &= w_I\ln\gamma_I.
\end{align}
Therefore~\eqref{eqn-opt-ineq} is shown for all $I\in\mathcal{I}$.
By the definition~\eqref{eqn-gamma}, we have
\begin{align}
    1 + \sum_{p(I)=\emp}\gamma_I
    = 1 + \sum_{I\in\mathcal{I}} \abs{\A(I)} = \abs{\Sigma}.
\end{align}
By the definition of $w_I$, we also have
\begin{align}
    1 + \sum_{p(I)=\emp} w_I
    = \max_{\bm x\in Q} \norm{\bm x}_1 = M_Q.
\end{align}
Finally, we obtained the following inequality that we wanted to show.
\begin{align}
    d^*(\bm 0) 
    &= \sum_{p(I)=\emp}\mathrm{opt}(I) \\
    &\le \sum_{p(I)=\emp}w_I\ln\gamma_I \\
    &\le \qty(
    1 + \sum_{p(I)=\emp}w_I
    )
    \ln\qty(
    1 + \sum_{p(I)=\emp}\gamma_I
    ) \\
    &= M_Q\ln\abs{\Sigma}.
\end{align}
\end{proof}

We conclude this section by proving the following theorem.
\begin{theorem}
    The prox-function $d$ defined by~\eqref{eqn-prox} is $1/(M_Q^2\ln\abs{\Sigma})$-strongly convex substantially, with respect to the $L_1$-norm.
    In other words, assume that $d$ is $\sigma$-strongly convex with respect to the $L_1$-norm, and let $D := \max_{\bm x\in Q} d(\bm x) - \min_{\bm x\in Q} d(\bm x)$, then the following inequality holds:
    \begin{align}
        \frac{D}{\sigma} \le M_Q^2\ln\abs{\Sigma}.
    \end{align}
\end{theorem}
\begin{proof}
    It follows from Theorem~\ref{thm-sigma} and Theorem~\ref{thm-d}.
\end{proof}
\section{Numerical experiments}
\label{sec-experiments}
In this section, we report the results of solving extensive-form games by using the prox-function \eqref{eqn-prox} with EGT.
As toy instances, we used three games, Kuhn poker, Leduc Hold'em (3 ranks), and Leduc Hold'em (13 ranks), which we explain in detail in Appendix \ref{sec-app-games}.
All implementations used in the experiments are available on~\url{https://github.com/habara-k/egt-on-efg}.

\subsection{Heuristics to accelerate the convergence of EGT}
The parameters and choices between $\mu_1$ and $\mu_2$ to shrink proposed in~\cite{nesterov2005excessive} guarantee the convergence of \eqref{eqn-conv} but are very conservative, then heuristic-based parameter selection will perform better in most cases.
First, we always use the following heuristics in~\cite{farina2021better}:
\begin{itemize}
    \item To start with small $\mu_1$ and $\mu_2$, we call the initializing algorithm \ref{alg-init} with $\mu_1=\mu_2=10^{-6}$.
    Increase $\mu_1$ and $\mu_2$ by 20\% until the output satisfies the excessive gap condition.
    \item In each step, shrink the larger between $\mu_1$ and $\mu_2$.
    \item To obtain a large step size, we call the shrinking algorithm \ref{alg-shrink} with the global $\tau$, which is initially set to 0.5 and is decreased by 50\% while the output does not satisfy the excessive gap condition.
\end{itemize}
See Algorithm \ref{alg-egt} for details.
In addition to these heuristics, this paper proposes a \textit{centering trick}, which is still not considered in the related literature to the best of our knowledge. 
For $\bm x^\prime\in \ri Q$, we can define the following prox-function:
\begin{align}
    \tilde{d}(\bm x; \bm x^\prime) :=
    d(\bm x) - d(\bm x^\prime) - \nabla d(\bm x^\prime)^\T \qty(
    \bm x - \bm x^\prime
    ),
    \label{eqn-prox-z}
\end{align}
which we call $\bm x^\prime$\textit{-centered} prox-function because $\argmin_{\bm x\in Q} \tilde{d}(\bm x; \bm x^\prime) = \bm x^\prime$ holds.
The centering trick uses $\tilde{d}(\bm x;\bm x^\prime)$ as a prox-function for EGT, where $\bm x^\prime$ is a solution obtained by some other method.
The reason for using this centered prox-function is to improve the accuracy of the smoothing approximation $f_{\mu_2}(\phi_{\mu_1})$ in EGT.
In fact, when the prox-function $d_2(d_1)$ takes the minimum value 0 in the optimal solution $y^*(x^*)$ of BSPP, the minimization of $f$ and $f_{\mu_2}$ (the maximization of $\phi$ and $\phi_{\mu_1}$) are equivalent, for any $\mu_2(\mu_1) > 0$.
In addition, we show that
\begin{align}
     \nabla \tilde{d}^*(\bm\xi; \bm x^\prime) 
     := \argmax_{\bm x\in Q} \qty{\bm\xi^\T \bm x - \tilde{d}(\bm x; \bm x^\prime)} 
     = \nabla d^*(\bm\xi + \nabla d(\bm x^\prime)),
\end{align}
so the cost required for the calculation is also $\order{\abs{\Sigma}}$.
Numerical experiments evaluate the performance of the centering trick.

\subsection{Results of the experiments}
Five methods are used: CFR~\cite{zinkevich2007regret}, CFR+~\cite{tammelin2014solving}, EGT, EGT-centering, and EGT-centering with CFR+.
EGT uses the prox-function defined in \eqref{eqn-prox}.
In EGT-centering, the first EGT is performed in 10\% of the total steps, and the remaining 90\% of the steps are performed in the second EGT using the prox-function centered on the solution obtained from the first EGT. 
EGT-centering with CFR+ is a variant of EGT-centering, using CFR+ instead of the first EGT.

From Figure \ref{fig-experiments},
note that EGT-centering and EGT-centering with CFR+ both use EGT and CFR+ for the first 10\% of iterations, respectively, so only the last 90\% of iterations are drawn.
First, since Kuhn poker is a very small game, EGT, which is unsuitable for larger games, performs similarly to CFR+. EGT with centering tricks using the solutions of these two methods, namely EGT-centering and EGT-centering with CFR+, converge faster than the other methods.
Second, since Leduc Hold'em is a larger game than Kuhn poker, EGT converges worse than CFR+. However, we see that EGT-centering performs as well as CFR+ and that EGT-centering with CFR+ converges better than pure CFR+.

These results show that EGT combined with CFR+ by the centering trick converges faster than pure CFR+ in large games where EGT alone performs worse than CFR+.
In addition, although we could not experiment in this paper, the results suggest that EGT with the centering trick has the potential to further exploit the performance of CFR+ in very large games such as those used in~\cite{bowling2015heads, brown2018superhuman}.
Note that the implementations of all five methods are optimized in the same way, resulting in EGT (including the centering trick) taking at most 2 to 3 times longer per iteration than CFR (CFR+).
Thus, although the horizontal axis in Figure~\ref{fig-experiments} represents the number of iterations, changing this to the running time (computational cost) yields almost the same result.

\begin{figure}[htbp]
    \begin{minipage}[b]{\columnwidth}
        \centering
        \includegraphics[width=120mm]{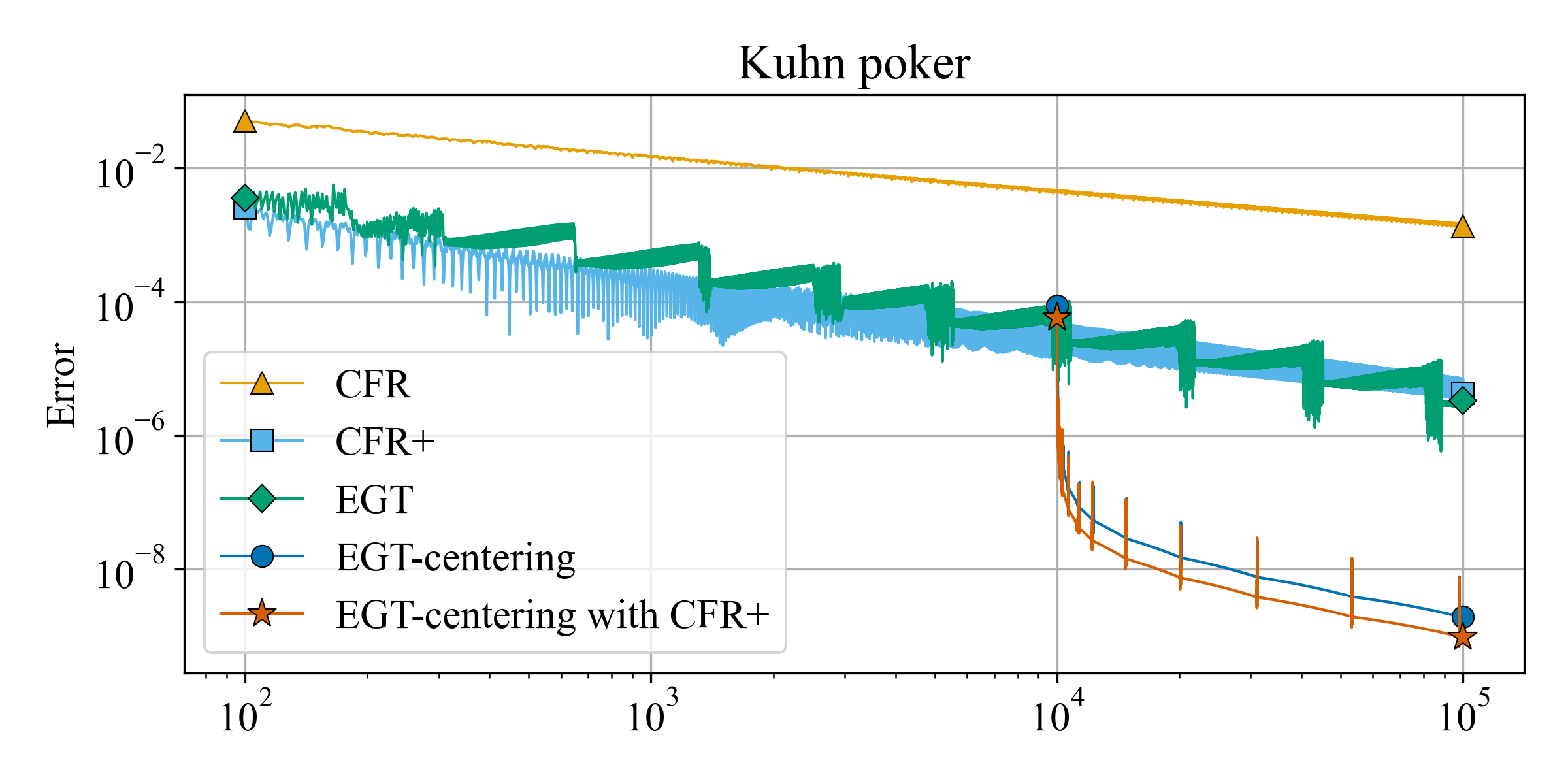}
    \end{minipage}
    \begin{minipage}[b]{\columnwidth}
        \centering
        \includegraphics[width=120mm]{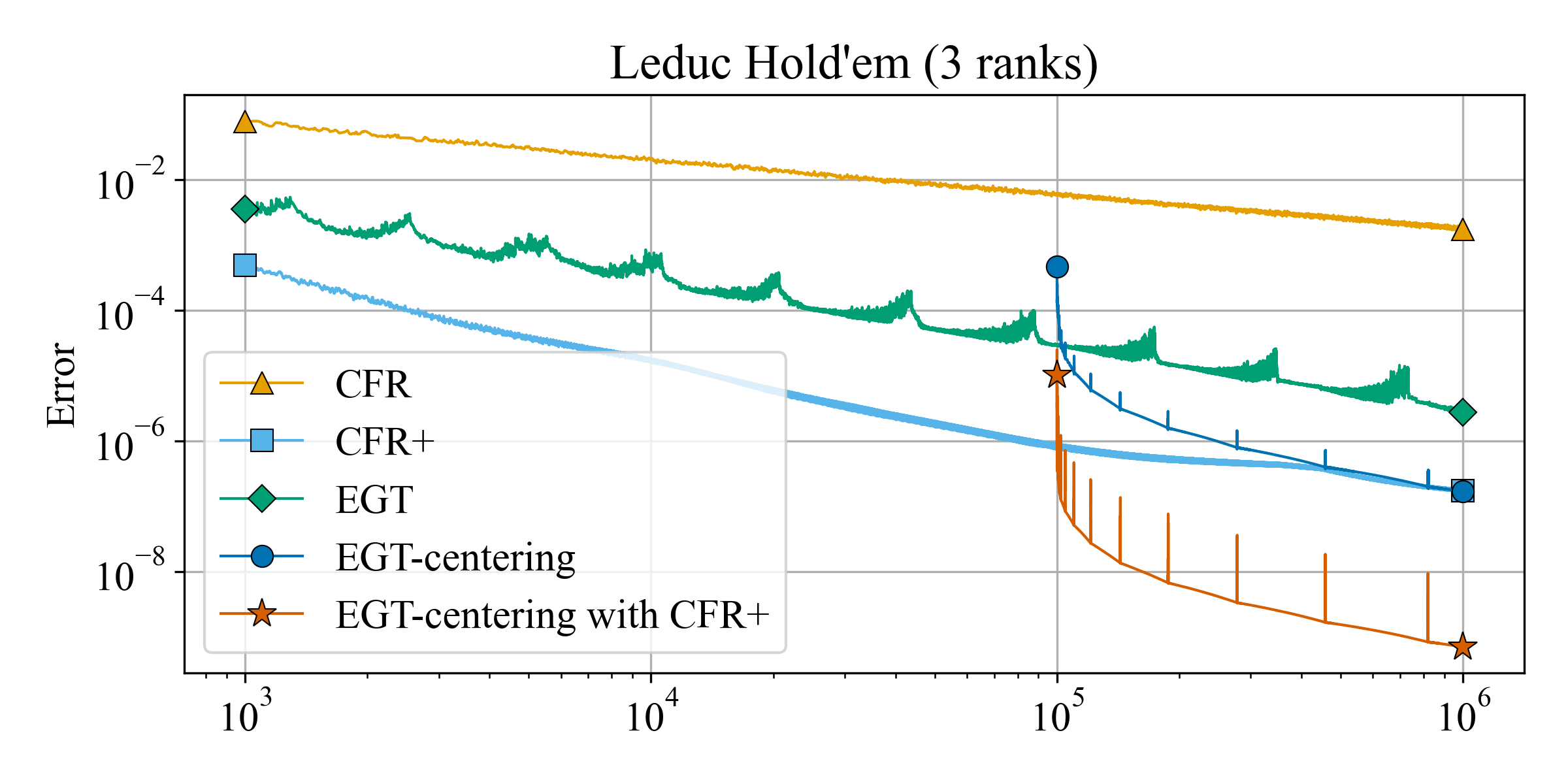}
    \end{minipage}
    \begin{minipage}[b]{\columnwidth}
        \centering
        \includegraphics[width=120mm]{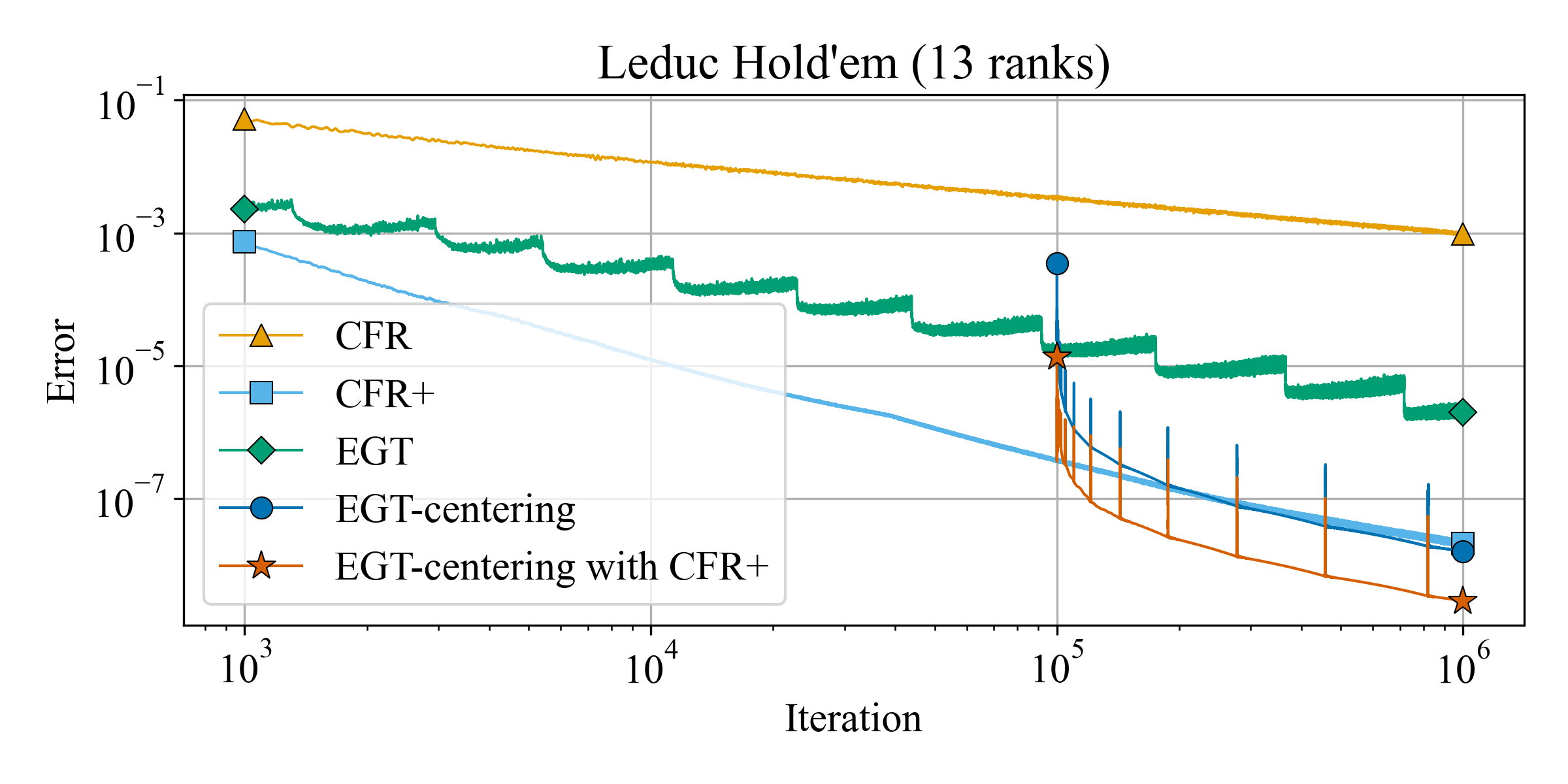}
    \end{minipage}
    \caption{Performance of the five methods for solving three games.}
    \label{fig-experiments}
\end{figure}

\begin{algorithm}[htbp]
\caption{Excessive gap technique (with heuristics in~\cite{farina2021better})}
\label{alg-egt}
\begin{algorithmic}[1]
\Function{init}{\null}
\State $\mu\gets 10^{-6}$
\While{true}
    \State $\bm x, \bm y \gets$ Call algorithm \ref{alg-init} with $\mu_1=\mu_2=\mu$.
    \If{$f_{\mu_2}(\bm x) \le \phi_{\mu_1}(\bm y)$}
        \State \Return $\bm x, \bm y, \mu, \mu$
    \EndIf
    \State $\mu \gets 1.2\mu$
\EndWhile
\EndFunction
\Statex
\Function{decrease$\mu_1$}{$\bm x, \bm y, \mu_1, \mu_2, \tau$}
    \While{true}
        \State $\bar{\bm x},\bar{\bm y} \gets $ Call algorithm \ref{alg-shrink} with $\bm x, \bm y, \mu_1, \mu_2,$ and $\tau$.
        \If{$f_{\mu_2}(\bar{\bm x}) \le \phi_{(1-\tau)\mu_1}(\bar{\bm y})$}
            \State \Return $\bar{\bm x}, \bar{\bm y}, (1-\tau)\mu_1, \tau$
        \EndIf
        \State $\tau \gets 0.5\tau$
    \EndWhile
\EndFunction
\Statex
\Ensure{$\bm x\in Q_1, \bm y\in Q_2$: solutions of EGT}
\State $\bm x, \bm y, \mu_1, \mu_2 \gets $\Call{init}{\null}
\State $\tau \gets 0.5$
\For{$k=1,\dots,T-1$}
    \If{$\mu_1 > \mu_2$}
        \State $\bm x, \bm y, \mu_1, \tau \gets $ \Call{decrease$\mu_1$}{$\bm x, \bm y, \mu_1, \mu_2, \tau$}
    \Else
        \State decrease $\mu_2$ similarly.
    \EndIf
\EndFor
\end{algorithmic}
\end{algorithm}

\section{Conclusions}
\label{sec-conclusions}
To make the smoothing method for solving extensive-form games applicable to large games, we proposed a modified version of the prox-function used internally in the smoothing method to improve its theoretical convergence guarantees.
We also proposed a heuristic to update the prox-function with temporary solutions, which is expected to improve the accuracy of the internal approximation of the smoothing method.
Numerical experiments confirm that the smoothing method applying the heuristic with the solution of CFR+, the state-of-the-art method for large games, converges faster than pure CFR+.
Some future works include confirming the performance of the smoothing method when applied to larger-scale extensive-form games and proving theoretically that the proposed heuristic improves convergence.

\bibliographystyle{unsrt}
\bibliography{bibfile}

\appendix

\section{Toy instances}
\label{sec-app-games}

In this section, we present some well-known extensive-form games.
These games are used in the numerical experiments in Section~\ref{sec-experiments}.

\subsection{Kuhn poker}
Kuhn poker is a simple variant of Texas Hold'em proposed in~\cite{kuhn1950simplified}, played with three cards, J, Q, and K.
Each player is dealt a card privately, and the game begins with each player betting 1 chip.
Player 1 is the first to act.
Player 1 can either \textit{check} (do nothing and pass the turn to Player 2) or \textit{raise} (bet an additional 1 chip).
%\begin{enumerate}[label=(\alph*)]
\begin{enumerate}[(a)]
    \item If player 1 raises, player 2 chooses to \textit{call} (bet an additional 1 chip) or \textit{fold} (surrender and lose the 1 chip already bet);
    if player 2 calls, the player with the higher-ranked card wins all chips (showdown).
    \label{itm-kuhn}
    \item If player 1 checks, player 2 chooses to check (and showdown) or raise;
    if player 2 raises, player 1 chooses to call or fold as in \ref{itm-kuhn}.
\end{enumerate}
The game tree of Kuhn poker is shown in Figure~\ref{fig-kuhn}, and its information tree of each player is shown in Figures~\ref{fig-kuhnq1} and~\ref{fig-kuhnq2}.

\begin{figure}[tbp]
    \centering
    \includegraphics[width=120mm]{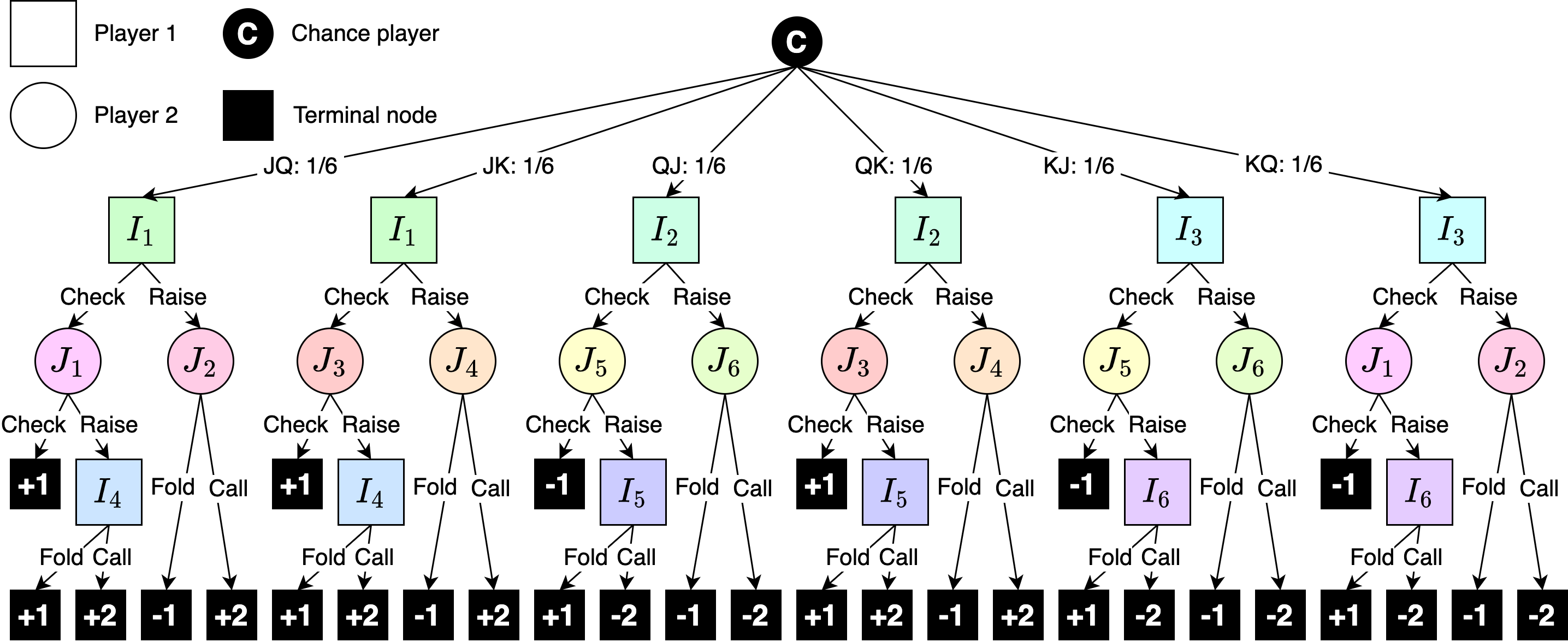}
    \caption{
    Kuhn poker game tree.
    The nodes marked with ``C'' on a black background are played by the chance player, and the square and circle nodes are played by player 1 and player 2.
    The terminal nodes are represented by a black square with corresponding $u(z)$.
    The arrows correspond to actions, and the state moves down to the pointed node.
    Chance player action probabilities are given and are written on the arrows.
    The nodes played by player 1 and player 2 are partitioned according to the information partition $\mathcal{I}_1=\qty{I_1, \dots, I_6}$ and $\mathcal{I}_2=\qty{J_1, \dots, J_6}$. 
    Each player does not know which card is dealt to their opponent, so they cannot distinguish between nodes belonging to the same information partition.
    }
    \label{fig-kuhn}
\end{figure}

\begin{figure}[htbp]
    \centering
    \includegraphics[width=100mm]{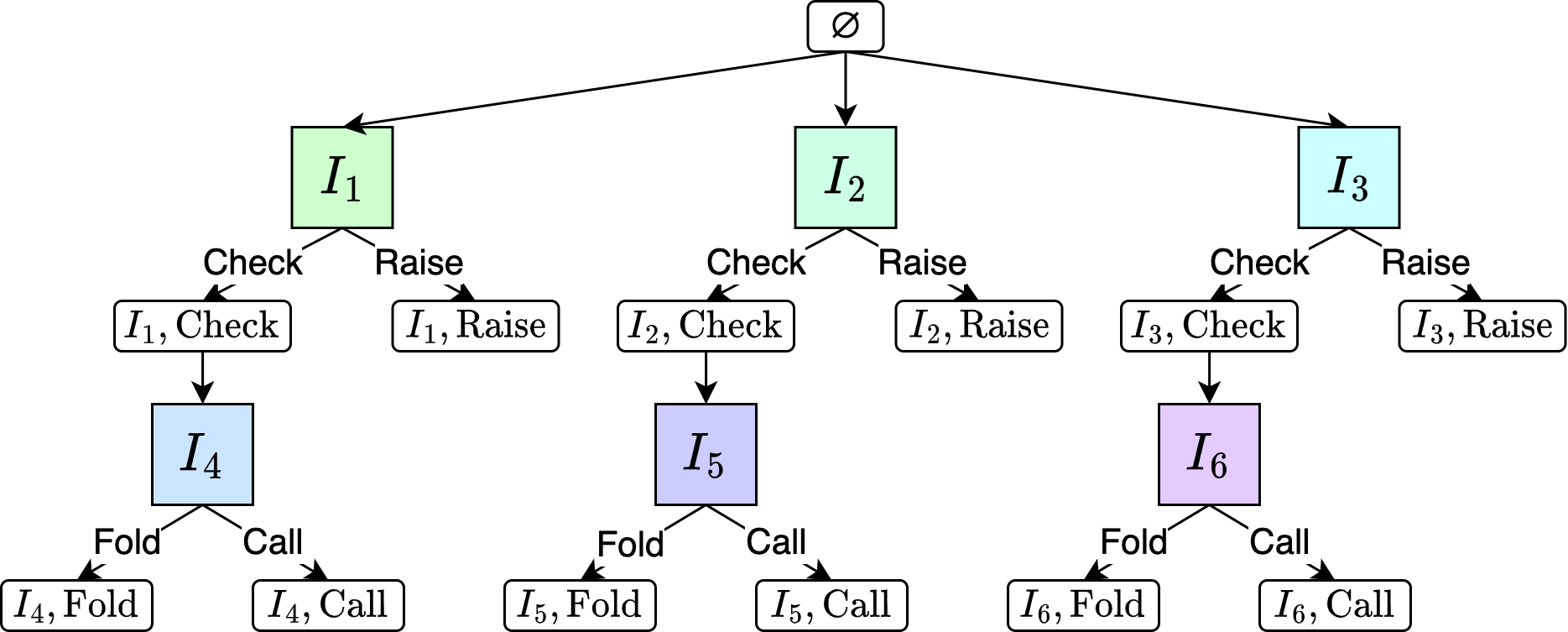}
    \caption{The \textit{information tree} of player 1 in Kuhn poker}
    \label{fig-kuhnq1}
\end{figure}

\begin{figure}[htbp]
    \centering
    \includegraphics[width=120mm]{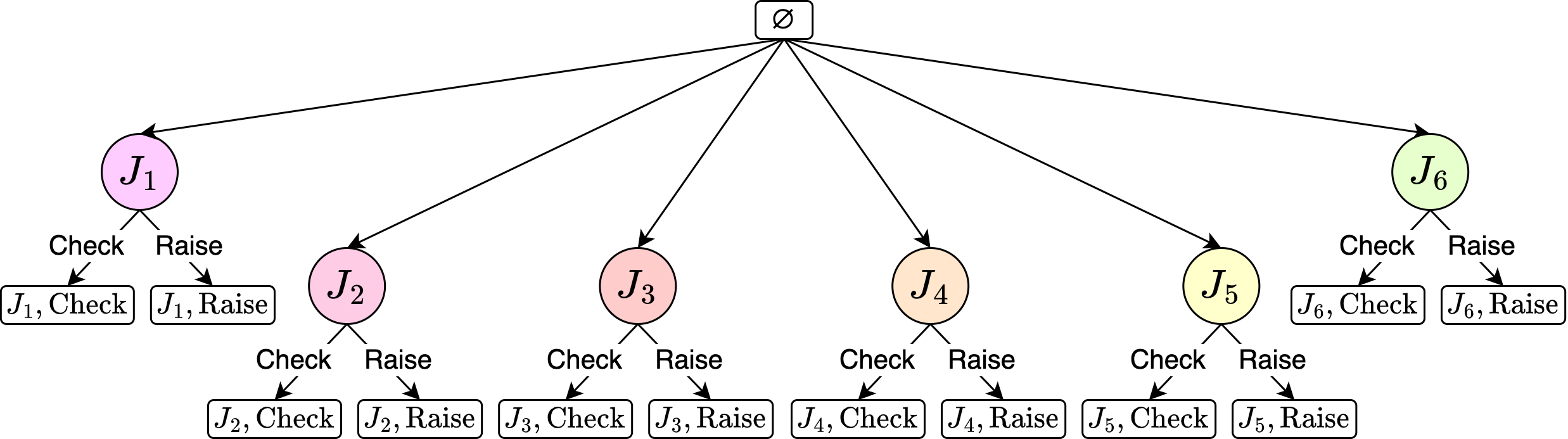}
    \caption{The \textit{information tree} of player 2 in Kuhn poker}
    \label{fig-kuhnq2}
\end{figure}

\subsection{Leduc Hold'em}
Leduc Hold'em is a simple variant of Texas Hold'em proposed in~\cite{southey2012bayes}, played with pairs of cards J, Q, and K, totaling six cards.
The general rules are the same as in Kuhn poker, but some differences exist.
\begin{itemize}
    \item Each player can raise against the opponent's raise, which is called \textit{re-raise}, but cannot raise against a re-raise.
    \item If both players check or one player calls, instead of an immediate showdown, a \textit{community card} is revealed randomly from the remaining cards, and the game is resumed only once more.
    \item The winner of the showdown is the player with the same card as the community card.
    If there is no such player, the player with the higher-ranked card wins.
    If both players have the same cards, the game is tied (chips are divided equally).
    \item The betting amount of raise is 2 and 4 chips in each phase.
    This means that the maximum move of chips is 1+2+2+4+4=13.
\end{itemize}
The game tree of Leduc hold'em is shown in Figure \ref{fig-leduc}, and the information tree of player 1 is shown in Figure \ref{fig-leducq1}.
Experiments are also performed for Leduc Hold'em (13 ranks), a game in which the number of card ranks is changed from 3 to 13.

\begin{figure}[htbp]
    \centering
    \includegraphics[width=125mm]{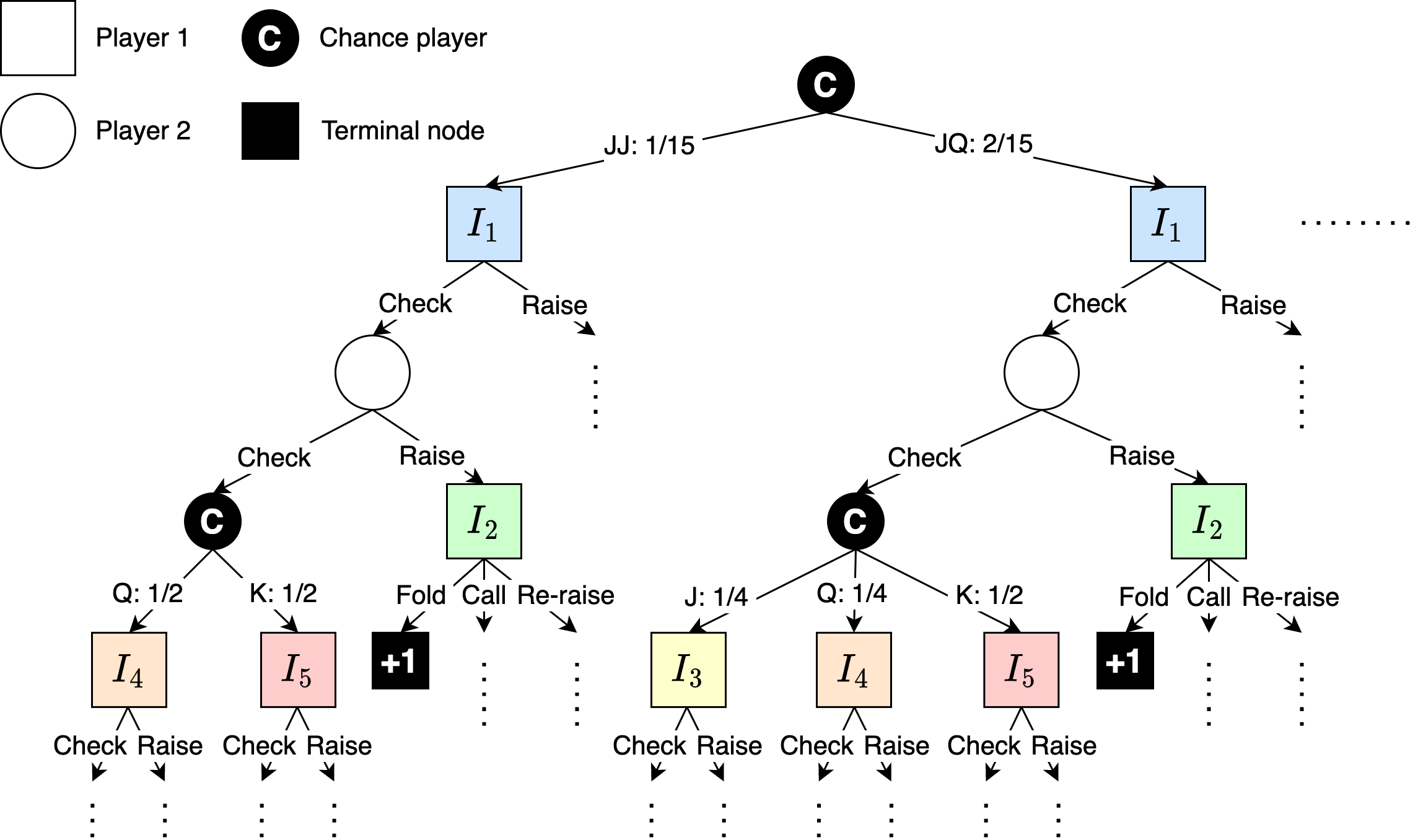}
    \caption{Leduc Hold'em game tree}
    \label{fig-leduc}
\end{figure}

\begin{figure}[htbp]
    \centering
    \includegraphics[width=120mm]{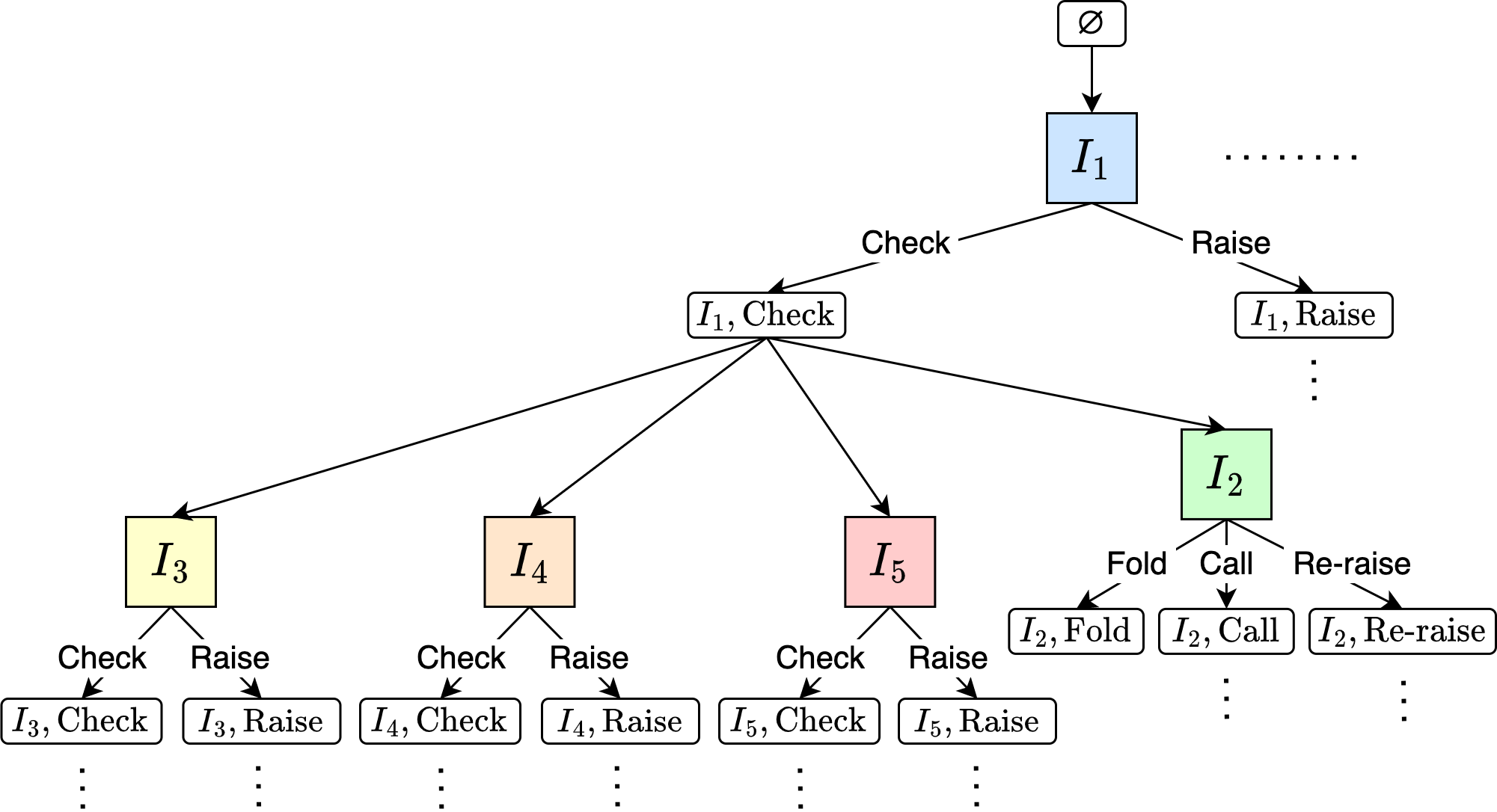}
    \caption{The \textit{information tree} of player 1 in Leduc Hold'em}
    \label{fig-leducq1}
\end{figure}
\section{Subproblem in the conjugate function}
\label{sec-app-subproblem}
Here, we will compute the analytical solution of the problem below, which is associated to our proposed prox-function.
Let us consider the following maximization problem:
\begin{maxi}
    {\bm z\in\R^n_{>0}}{\sum_{i=1}^n \xi_i z_i - w\sum_{i=1}^n z_i\ln z_i}{}{}
    \label{eqn-app-opt}
    \addConstraint{\sum_{i=1}^n z_i}{=1,}
\end{maxi}
where $w > 0$. Lagrange function $L\colon\R^n\times\R\to\R$ is defined by
\begin{align}
    L(\bm z, \mu) := \sum_{i=1}^n \xi_i z_i - w\sum_{i=1}^n z_i\ln z_i + \mu\qty(
    1 - \sum_{i=1}^n z_i
    ),
\end{align}
and the optimal solution $\bm z^*$ of \eqref{eqn-app-opt} satisfies the following equation with some $\mu^*\in\R$:
\begin{align}
    \pdv{L}{z_i}\qty(z^*,\mu^*)
    &= \xi_i - w(1 + \ln z_i^*) - \mu^* = 0, 
    \quad i=1,\dots,n
    \\
    \pdv{L}{\mu}\qty(z^*,\mu^*) 
    &= 1 - \sum_{j=1}^n z_i^* = 0.
\end{align}
Eliminate $\mu^*$ from the above equation to obtain
\begin{align}
    z_i^* = \frac{\exp(\xi_i/w)}{\sum_{j=1}^n \exp(\xi_j/w)},
\end{align}
and the optimal value is given by
\begin{align}
    \sum_{i=1}^n \xi_i z_i^* - w\sum_{i=1}^n z_i^*\ln z_i^*
    &= 
    \sum_{i=1}^n \xi_i z_i^* - w\sum_{i=1}^n z_i^*\qty{(\xi_i/w) - \ln\sum_{j=1}^n \exp(\xi_j/w)} \\
    &= w\ln\sum_{j=1}^n \exp(\xi_j/w).
\end{align}
\section{Implementation of CFR and CFR+}

Most implementations of CFR and CFR+ employ the depth-first search of the game tree~\cite{tammelin2014solving,neller2013introduction}.
This method has the advantage of greatly reducing computation time in huge games, at the expense of accuracy, by using sampling.
On the other hand, simply exploring all nodes may be computationally inefficient.
For example, in a game with multiple terminal nodes $z \in Z$ such that the pair $(p_1(z),p_2(z))$ is equal (see Section \ref{sec-efg} for the definition of $p_1, p_2$), all terminal nodes are searched at each iteration in the depth-first search type implementation, but in the BSPP type unnecessary. 

This paper presents CFR and CFR+ implementations using the BSPP representation of the extensive-form game introduced in Section \ref{sec-efg}.
These implementations are also available on~\url{https://github.com/habara-k/egt-on-efg}.

\begin{algorithm}[htbp]
\caption{CFR and CFR+ (in BSPP form)}
\begin{algorithmic}[1]
\Function{prod}{$\bm z\in\R_{\ge 0}^\abs{\Sigma}$}
\State $\bm x \gets \bm 0\in\R^\abs{\Sigma}$
\State $x_\emp \gets 1$
\For{$I\in\mathcal{I}$ in top-down order}
    \For{$a\in\A(I)$}
        \State $x_{I,a} \gets x_{p(I)} z_{I,a}$
    \EndFor
\EndFor
\State \Return $\bm x$
\EndFunction
\Statex
\Function{normalize}{$\bm r\in\R_{\ge 0}^\abs{\Sigma}$}
\State $\bm z\gets \bm 0\in\R^\abs{\Sigma}$
\For{$I\in\mathcal{I}$}
    \If{$\exists a\in\A(I) \text{ s.t. } r_{I,a}>0$}
        \For{$a\in\A(I)$}
            \State $z_{I,a} \gets r_{I,a}/\sum_{a\in\A(I)} r_{I,a}$
        \EndFor
    \Else
        \For{$a\in\A(I)$}
            \State $z_{I,a} \gets 1/\abs{\A(I)}$
        \EndFor
    \EndIf
\EndFor
\State \Return $\bm z$
\EndFunction
\Statex
\Function{regret}{$\bm z\in\R^\abs{\Sigma}, \bm u\in\R^\abs{\Sigma}$}
\State $\bm r\gets \bm 0\in\R^\abs{\Sigma}$
\For{$I\in\mathcal{I}$ in bottom-up order}
    \For{$a\in\A(I)$}
        \State $r_{I,a} \gets u_{I,a} - \sum_{a\in\A(I)} u_{I,a} z_{I,a}$
    \EndFor
    \State $u_{p(I)} \gets u_{p(I)} + \sum_{a\in\A(I)} u_{I,a} z_{I,a}$
\EndFor
\State \Return $\bm r$
\EndFunction
\algstore{cfr}
\end{algorithmic}
\end{algorithm}

\begin{algorithm}[htbp]
\begin{algorithmic}[1]
\algrestore{cfr}
\Ensure{$\bm x\in Q_1, \bm y\in Q_2$: solutions of CFR}
\State $\bm r_\X, \bm r_\Y \gets \bm 0\in\R^\abs{\Sigma_1}, \bm 0\in\R^\abs{\Sigma_2}$
\State $\bm z_\X, \bm z_\Y \gets $ \Call{normalize}{$\bm r_\X$}, \Call{normalize}{$\bm r_\Y$}
\State $\bm x^1, \bm y^1 \gets $ \Call{prod}{$\bm z_\X$}, \Call{prod}{$\bm z_\Y$}
\For{$k=1,\dots,T-1$}
    \State $\bm r_\X \gets \bm r_\X \ + $ \Call{regret}{$\bm z_\X, -\bm A\bm y^k$}
    \State $\bm r_\Y \gets \bm r_\Y \ + $ \Call{regret}{$\bm z_\Y, \bm A^\T\bm x^k$}
    \State $\bm z_\X \gets $ \Call{normalize}{$\max(\bm r_\X, \bm 0)$}
    \State $\bm z_\Y \gets $ \Call{normalize}{$\max(\bm r_\Y, \bm 0)$}
    \State $\bm x^{k+1} \gets $ \Call{prod}{$\bm z_\X$}
    \State $\bm y^{k+1} \gets $ \Call{prod}{$\bm z_\Y$}
\EndFor
\State $\bm x, \bm y \gets \frac{1}{T}\sum_{k=1}^T k\bm x^k, \frac{1}{T}\sum_{k=1}^T k\bm y^k$
\Statex
\Ensure{$\bm x\in Q_1, \bm y\in Q_2$: solutions of CFR+}
\State $\bm r_\X, \bm r_\Y \gets \bm 0\in\R^\abs{\Sigma_1}, \bm 0\in\R^\abs{\Sigma_2}$
\State $\bm z_\X, \bm z_\Y \gets $ \Call{normalize}{$\bm r_\X$}, \Call{normalize}{$\bm r_\Y$}
\State $\bm x^1, \bm y^1 \gets $ \Call{prod}{$\bm z_\X$}, \Call{prod}{$\bm z_\Y$}
\For{$k=1,\dots,T-1$}
    \State $\bm r_\X \gets \max(\bm r_\X \ + $ \Call{regret}{$\bm z_\X, -\bm A\bm y^k$}$, \bm 0)$
    \State $\bm z_\X \gets $ \Call{normalize}{$\bm r_\X$}
    \State $\bm x^{k+1} \gets $ \Call{prod}{$\bm z_\X$}
    
    \State $\bm r_\Y \gets \max(\bm r_\Y \ + $ \Call{regret}{$\bm z_\Y, \bm A^\T\bm x^{k+1}$}$, \bm 0)$
    \State $\bm z_\Y \gets $ \Call{normalize}{$\bm r_\Y$}
    \State $\bm y^{k+1} \gets $ \Call{prod}{$\bm z_\Y$}
\EndFor
\State $\bm x, \bm y \gets \frac{2}{T+T^2}\sum_{k=1}^T k\bm x^k, \frac{2}{T+T^2}\sum_{k=1}^T k\bm y^k$
\end{algorithmic}
\end{algorithm}

\end{document}